\documentclass[11pt]{article}
\usepackage{defs}

\title{Improved Quantum Algorithms for Subset Sum and $k$-SUM}
\author{
   Nikolai Chukhin
   \thanks{JetBrains Research. Email: \url{buyolitsez1951@gmail.com}}
   \and
   Alexander S. Kulikov
   \thanks{JetBrains Research. Email: \url{alexander.s.kulikov@gmail.com}}
   \and
   Maksim Levitskii
   \thanks{Neapolis University Pafos and JetBrains Research. Email: \url{max.levitsky.5@gmail.com}}
   \and
   Ivan Mihajlin
   \thanks{JetBrains Research. Email: \url{ivmihajlin@gmail.com}}
}
\date{}

\begin{document}

\maketitle

\begin{abstract}
	The \problemname{Subset Sum} problem asks whether, given $n$~integers and a~target, some subset of~the~integers sums to~the~target.
	Its best known worst-case running time is~$O^*(2^{n/2})$ (Horowitz and Sahni, 1974), whereas the~best quantum upper bound is~$O^*(2^{n/3})$ (Bernstein, Jeffery, Lange, and Meurer, 2013).
	The $k$-\problemabbr{SUM} problem is~a~parameterized version of~\problemname{Subset Sum} asking whether there are $k$~integers that sum
	to~the target.
	The best classical upper bound for~it~is~$\widetilde O(n^{\lceil k/2\rceil})$, whereas the best quantum running time is~$\widetilde O(n^{k/3})$ (Tani, 2009).
	For random instances, a~quantum algorithm with running time $\widetilde O(n^{\Phi_k})$ is~known, where
	\[\Phi_k=\frac{2k-\lfloor k/7\rfloor-\lfloor (k+3)/7\rfloor}{6}\]
	(Schrottenloher, 2021).

	We~present a~new quantum algorithm solving worst-case $k$-\problemabbr{SUM} in~time
	$\widetilde O(n^{\Psi_k})$, where
	\[\Psi_k = \Phi_k -\frac{[k\equiv3\bmod 7]}9-\frac{[k\equiv6\bmod 7]}{18}.\]
	The algorithm is~not only faster for all~$k$ congruent to~$3$ or~$6$ modulo~$7$, but also gives a~worst-case guarantee rather than a~guarantee restricted to~single-solution random instances.
	Combining our algorithm for $7$-\problemabbr{SUM} with the standard block reduction technique yields an~$O^*(2^{2n/7})$ quantum algorithm for \problemname{Subset Sum}, improving the previously known $O^*(2^{n/3})$ algorithm.
\end{abstract}

\clearpage
\section{\texorpdfstring{Complexity of~Subset Sum and $k$-SUM}{Complexity of~Subset Sum and k-SUM}}
\label{sec:introduction}

The \problemname{Subset Sum} problem asks whether, given $n$ integers and a~target $t$, some subset of~the~integers sums to~$t$.
Despite its elementary formulation, \problemname{Subset Sum} is~one of~Karp's original 21 \cc{NP}-complete problems~\cite{Karp72} and has been studied extensively in~exact algorithms~\cite{HS74,SS81,DDKS12,Bringmann17,RW26} and cryptography~\cite{MH78,LO85,CR88,IN96,LPS10,FMV16,JW26}.

\paragraph{Classical complexity of~Subset Sum.}
A~classical algorithm due to~Horowitz and Sahni~\cite{HS74} solves \problemname{Subset Sum} as~follows: it~partitions the~input into two halves, enumerates all subset sums of~each half, and searches for a~pair of~sums, one from each list, whose total sum equals $t$, thereby achieving $O^*(2^{n/2})$ time and space\footnote{$O^*(\cdot)$ suppresses factors that grow polynomially with the~input length, while $\widetilde O(\cdot)$ suppresses polylogarithmic factors.}.
The~best known worst-case running time improves this bound only by~a~polynomial factor, to~$O(2^{n/2}n^{-\gamma})$ for a~constant $\gamma>0.5023$~\cite{CJRS23}.
Designing an~$O^*(2^{(1 / 2 - \varepsilon) n})$ algorithm for \problemname{Subset Sum} is~a~well-known challenge.
For random instances of~\problemname{Subset Sum}, the~classical running time was reduced to~$O^*(2^{0.283n})$~\cite{BBSS20}, following the line of~work~\cite{HJ10,BCJ11}.
More recently, Li~\cite{Li25} proposed a~heuristic ternary-representation algorithm running in~time $O^*(2^{0.2400n})$.
The~space usage of~the~Horowitz--Sahni~\cite{HS74} algorithm can, however, be~reduced substantially without changing the~exponential running time.
Schroeppel and Shamir retained $O^*(2^{n/2})$ time while reducing the~space to~$O^*(2^{n/4})$~\cite{SS81}.
Nederlof and W\k{e}grzycki subsequently lowered the~space exponent to~$0.249999$~\cite{NW21}, and it~was later reduced to~$0.246$ while preserving the~same running time~\cite{BCKM24}.

\paragraph{Quantum complexity of~Subset Sum.}
For arbitrary inputs, one can partition the~input into three blocks, form one list of~subset sums for each block, and then use quantum search to~find a~pair of~entries from two lists whose sum occurs in~the~third list, yielding a~running time of~$O^*(2^{n/3})$~\cite{BJLM13}.
A~different construction attains the~same asymptotic bound~\cite{AHJKS22}.
For random instances, the~quantum running time was lowered to~$O^*(2^{0.218n})$~\cite{BBSS20}, following~\cite{BJLM13,HM18}.
Li~\cite{Li25} also proposed a~heuristic quantum algorithm running in~time $O^*(2^{0.1843n})$.

\paragraph{Classical complexity of~$k$-SUM.}
An~instance of~$k$-$\problemabbr{SUM}$ consists of~$k$ lists $A_1,\ldots,A_k$, each of~length $n$, and a~target $t$; the~task is~to~choose one entry from each list so~that their sum equals $t$.\footnote{Another common formulation gives a~single list $A$ of~$n$ integers and asks whether there exist $k$ distinct indices $i_1,\ldots,i_k$ such that $A[i_1]+\cdots+A[i_k]=t$.
For fixed~$k$, the two versions have roughly the same complexity~\cite{LWWW16}: one of~them can be~solved in~time $\widetilde{O}(n^{\alpha})$ if~and only if~the other can be~solved in~the same time.}
A~standard way to~solve \problemname{Subset Sum} is~to~reduce it~to~$k$-\problemabbr{SUM} for some value of~$k$.
The~standard block reduction partitions the~$n$ input integers into $k$ blocks and lists all subset sums within each block, producing a~$k$-$\problemabbr{SUM}$ instance whose lists have length~$O^*(2^{n/k})$.
The~meet-in-the-middle algorithm of~Horowitz and Sahni~\cite{HS74} is~precisely the~case $k=2$, and the~more elaborate algorithms in~\cite{SS81,DDKS12,NW21,BCKM24} reduce \problemname{Subset Sum} to~$k$-\problemabbr{SUM} for larger values of~$k$.

The~folklore meet-in-the-middle algorithm for $k$-\problemabbr{SUM} partitions the~$k$ lists into two groups, enumerates all sums obtained by~selecting one entry from each list in~a~group, and searches for a~pair of~complementary sums, yielding the~running time $\widetilde O(n^{\lceil k/2\rceil})$.
The~$k$-$\problemabbr{SUM}$ Conjecture~\cite{Patrascu10,AL13} asserts that this algorithm is~essentially optimal.
The conjecture is~known to~imply lower bounds for various computational problems~\cite{GO95,HB96,Patrascu10,AL13,WW13,AW14,JV16,KPP16}.
Moreover, its connection to~other fine-grained conjectures has recently been established~\cite{BKMRRS24,Lampis26}.

\paragraph{Quantum complexity of~$k$-SUM.}
For $k=2$, Ambainis~\cite{Ambainis07}, based on~the result of~\cite{BDHHMSW05}, gave an~$\widetilde O(n^{2/3})$ quantum algorithm.
For all values of~$k$, the upper bound $\widetilde O(n^{k/3})$ can be~obtained via the \problemname{Claw Finding} algorithm of~Tani~\cite{Tani09}.
For $k=4$, an~algorithm for random instances running in~time $\widetilde O(n^{6/5})$ was obtained in~\cite{BJLM13}.
A~classical foundation for this line of~work on~random instances is~Wagner's merging-tree algorithm for the~\problemname{Generalized Birthday Problem}~\cite{Wagner02}.
\emph{Random instances} of~$k$-\problemabbr{SUM} were studied in~\cite{GNS18,NS20,Schrottenloher21}.
For single-solution random instances, the~best known $k$-\problemabbr{SUM} algorithm is~due to~Schrottenloher~\cite{Schrottenloher21}\footnote{Some of~the bounds above were proven for the $k$-\problemabbr{XOR} problem, but they carry over to~$k$-\problemabbr{SUM}.}.

\begin{theorem}[\cite{Schrottenloher21}]
	For every integer $k>2$, single-solution random instances of~$k$-$\problemabbr{SUM}$ can be~solved in~quantum time $\widetilde O(n^{\Phi_k})$, where
	\[\Phi_k=\frac{2k-\lfloor k/7\rfloor-\lfloor (k+3)/7\rfloor}{6}.\]
\end{theorem}
Here, \emph{single-solution random instances} are those for which the~expected number of~witnesses is~one; see~\cite{Schrottenloher21} for the~exact definition.

\subsection*{Our Results}
Our main result is~the quantum algorithm for the \emph{worst-case} $k$-$\problemabbr{SUM}$ problem.

\begin{restatable}{theorem}{thmmain}
\label{thm:four-block}
For every integer $k>3$, worst-case $k$-$\problemabbr{SUM}$ can be~solved in~quantum time $\widetilde O(n^{\Psi_k})$, where\footnote{$[\cdot]$ is~the Iverson bracket: for a~condition $\mathcal{E}$, $[\mathcal{E}] = 1$ if~$\mathcal{E}$ holds, and $[\mathcal{E}] = 0$ otherwise.}
\[
	\Psi_k = \underbrace{\frac{2k-\lfloor k/7\rfloor-\lfloor (k+3)/7\rfloor}{6}}_{\Phi_k}-\frac{[k\equiv3\bmod 7]}9-\frac{[k\equiv6\bmod 7]}{18}.
\] 
\end{restatable}
This improves Tani's $\widetilde O(n^{k/3})$ algorithm for every value of~$k > 3$, and, asymptotically, $\Psi_k = \frac{2k}{7} + O(1)$.

A~simple corollary of~our $k$-\problemabbr{SUM} algorithm is~a~new upper bound for the \problemname{Subset Sum} problem, which improves the~previous worst-case quantum bound $O^*(2^{n/3})$~\cite{BJLM13,AHJKS22}.
\begin{corollary}
	The~\problemname{Subset Sum} problem can be~solved in~quantum time $O^*(2^{2n/7})$.
\end{corollary}

Beyond~\problemname{Subset Sum}, our $k$-\problemabbr{SUM} algorithm also improves the~quantum upper bound for a~related variant.
In~\problemname{Pigeonhole Equal Subset Sum}, one seeks two subsets with equal sums; here, the~$n$ input integers are promised to~sum to~less than $2^n-1$, so~that two such subsets exist by~the pigeonhole principle.
An~$O^*(2^{n/2})$ classical algorithm for this problem~\cite{AHJKS22} was improved by~Jin and Wu~\cite{JW24} to~$O^*(2^{0.4n})$, and subsequently by~Jin, Williams, and Zhang to~$O^*(2^{n/3})$~\cite{JWZ25}.
No~algorithm faster than $O^*(2^{n/3})$ is~known for the \problemname{Pigeonhole Equal Subset Sum} problem, even in~the quantum setting.
These works also raised the question of~whether the \emph{modular} version of~this problem can be~solved faster than $O^*(2^{0.5n})$.
We~give a~positive answer in~the quantum setting, showing that the \problemname{Pigeonhole Modular Equal Subset Sum} problem can be~solved in~quantum time $O^*(2^{0.453n})$ by~a~simple reduction to~the $k$-\problemabbr{SUM} problem and an~application of~\Cref{thm:four-block}.
Formal definitions and proofs are provided in~\Cref{thm:pigeonhole-modular-equal-sums} of~\Cref{sec:pigeonhole-modular-equal-subset-sum}.

\paragraph{Techniques and comparison with previous works.}
Schrottenloher's algorithm uses a~merging tree~\cite{Wagner02,NS20,Schrottenloher21}.
For random instances, it~recursively combines partial sums and uses modular constraints to~keep the~intermediate lists small.
At~the end, it~uses \problemname{Claw Finding} to~find a~solution.
Consequently, the~analysis of~both the~success probability and the~efficiency depends on~the~partial sums being sufficiently well distributed.
Our algorithm replaces this merging-tree paradigm with a~single decomposition into four blocks and one modular filtering step based on~a~random prime.
We~then use a~quantum walk to~reduce the~search space, quantum search to~find a~suitable residue, and, at~the very end, \problemname{Claw Finding} to~find a~solution.

Below, we~describe the main ideas underlying our algorithm in~greater detail.
We~group $k$ input lists into four consecutive blocks of~sizes approximately $\frac{3}{14}k$, $\frac{4}{14}k$, $\frac{3}{14}k$, and $\frac{4}{14}k$.
A~tuple in~block~$i$ records one choice from each list in~that block; we~write $\mathcal X_i$ for the~set of~these tuples and $\Sigma_i$ for their sums.
This transforms the~problem into an~unbalanced $4$-\problemabbr{SUM} instance: the~first and third grouped lists are smaller, whereas the~second and fourth are too large to~enumerate.

Instead of~constructing the~two larger lists explicitly, the~algorithm stores $m$-element subsets $S_2\subseteq\mathcal X_2$ and $S_4\subseteq\mathcal X_4$, where $m$ is a~carefully chosen number that is a~bit smaller than $\min \{ |\mathcal{X}_{2}|, |\mathcal{X}_{4}| \}$.
However, $m$ is chosen sufficiently small so that there is time to examine $S_2$ and $S_4$.
We~then construct a~checker for a~fixed pair of~subsets $S_2, S_4$.
Hence, using quantum search, one can find suitable subsets $S_2, S_4$ using roughly $\sqrt{\frac{|\mathcal{X}_2|\cdot|\mathcal{X}_{4}|}{m^2}}$ calls to~the~checker.

Before initiating the~quantum search, the~algorithm chooses a~random prime $p = \widetilde{O}(m)$.
Within the~checker, the~algorithm groups the~sums from~$S_2$ and~$S_4$ according to~their residues modulo~$p$ and stores one canonical representative of~every nonempty residue bucket in~the~dictionary.
Intuitively, a~suitable random prime isolates the~relevant exact sums: with constant probability, the~block-$2$ and block-$4$ sums of~a~fixed solution are canonical and can therefore be~detected by~our algorithm.

If~the~checker terminates without finding a~solution involving tuples from~$S_2$ and~$S_4$, the~entire dictionary must be~recomputed for a~new pair of~subsets $S_2'$ and~$S_4'$.
This recomputation is~costly, so~we~use a~quantum walk on~pairs $(S_2, S_4)$ instead of~the quantum search.
Informally, the advantage of~the~quantum walk over quantum search in~this setting is~that the~quantum walk changes the~state $(S_2, S_4)$ to~a~\emph{neighbouring} state $(S_2', S_4')$, allowing the~dictionary to~be~recomputed efficiently.

To~check a~quantum state $(S_2, S_4)$, we~use a~quantum search over the~possible residue $q \in \mathbb{Z}_{p}$ of~the~left-hand sum $\Sigma_1+\Sigma_2$; the~right-hand residue is~then forced to~be~$-q$ (assuming that the target is~equal to~zero).
For a~fixed $q$, we~iterate over all elements $\alpha_1 \in \mathcal{X}_{1}$ and store the~value of~$\Sigma_1(\alpha_1)+\Sigma_2(\alpha_2)$, where $\alpha_2$ is~the~canonical tuple from~$S_2$ with the~compatible residue, that is, such that $\Sigma_1(\alpha_1)+\Sigma_2(\alpha_2)\equiv q\bmod p$.
The~right collection is~indexed by~$\alpha_3\in\mathcal X_3$ and contains the~analogous values $-\Sigma_3(\alpha_3)-\Sigma_4(\alpha_4)$ obtained from~$S_4$.
We~then use the~\problemname{Claw Finding} procedure, which searches for an~equal exact value in~the~two collections.
If~no~claw is~found, the~walk moves to~a~neighbouring state and updates its data.

\section{Preliminaries}
\label{sec:preliminaries}

For positive integers $n$ and $a$, $[n]=\{1,\ldots,n\}$ and $[n]^a$ denotes the~set of~all $a$-tuples of~elements of~$[n]$.
For a~list $A$, we~write $A[i]$ for its $i$-th~entry.
For a~tuple or~vector $u$, we~write $u[i]$ for its $i$-th~coordinate.
For a~finite set $\mathcal X$, $x\in_R\mathcal X$ denotes that $x$ is~chosen uniformly at~random from~$\mathcal X$.
For a~condition $\mathcal E$, the~Iverson bracket $[\mathcal E]$ equals $1$ if~$\mathcal E$ holds and $0$ otherwise.
All logarithms are base two.

\subsection{\texorpdfstring{The Subset Sum and $k$-SUM Problems}{The Subset Sum and k-SUM Problems}}
\label{sec:k-sum-problem}

\begin{definition}
	An~instance of~\problemname{Subset Sum} consists of~a~list $A$ of~$n$ integers and a~target $t\in\mathbb Z$.
	The~task is~to~determine whether there exists a~subset $S\subseteq[n]$ such that
	\[\sum_{i\in S}A[i]=t.\]
\end{definition}

\begin{definition}
	An~instance of~$k$-$\problemabbr{SUM}$ consists of~$k$ lists $A_1,\ldots,A_k$ of~integers, each of~size $n$, and a~target $t\in\mathbb Z$.
	The~task is~to~find $(i_1,\ldots,i_k)\in[n]^k$ satisfying
	\[A_1[i_1]+\cdots+A_k[i_k]=t,\]
	or~to~report that no~such tuple exists.
\end{definition}
It~is~common to~assume in~$k$-\problemabbr{SUM} that the target $t$ is~zero, which is~achieved by~replacing every $A_1[i]$ with $A_1[i]-t$.
It~is~also common to~assume that the bit-length of~all integers in~the $k$-\problemabbr{SUM} problem is~at~most $k \log n + O(\log \log n)$, that is, that every integer has absolute value $\widetilde O(n^k)$; this can be~achieved by~the standard fingerprinting technique~\cite{ALW14}.

\subsection{Prime Numbers}
\label{sec:prime-numbers}

Given an~integer~$M\geq2$, one can generate a~uniformly random prime from~$\{M,\ldots,2M\}$ in~expected time $\log^{O(1)}M$ as~follows.
One selects a~uniformly random integer from~$\{M,\ldots,2M\}$ and tests it~for primality in~time $\log^{O(1)}M$~\cite{Lenstra02,AKS04}.
By~the prime number theorem~\cite{Hadamard96,DeLaValleePoussin96}, the expected number of~trials is~$O(\log M)$.
We~also use the~following estimate.

\begin{proposition}
\label{prop:prime-divisor}
	Let $M\geq2$ and $z\neq0$ be~integers.
	The~probability that a~uniformly random prime from~$\{M,\ldots,2M\}$ divides~$z$ is
	\[O\!\left(\frac{\log|z|}{M}\right).\]
\end{proposition}
\begin{proof}
	Suppose that $d$ distinct primes from~$\{M,\ldots,2M\}$ divide~$z$.
	Their product divides~$|z|$, and hence $M^d\leq|z|$, which implies $d\leq\log|z|/\log M$.
	By~the~prime number theorem~\cite{Hadamard96,DeLaValleePoussin96}, the~interval contains $\Omega(M/\log M)$ primes, and therefore the~probability in~question is~$O(\log|z|/M)$.
\end{proof}

\subsection{Quantum Primitives}

Our algorithm employs the~standard circuit model augmented with coherent read--write quantum random-access memory, denoted \cc{QRAQM}.
We~count each single- or~two-qubit gate as~one operation and assume that coherent random access to~any entry of~a~list represented using $m$ qubits can be~performed in~time polylogarithmic in~$m$.
See~\cite{NS20, Schrottenloher21} for a~formal definition.

We~use the following standard primitives in~our algorithm.
\paragraph{Coherent dictionaries.}
\begin{theorem}[{\cite[Section~6.2]{Ambainis07}}]
\label{thm:coherent-dictionary}
	Let $M=n^{O(1)}$, and let $\mathcal U$ be~a~totally ordered universe whose elements have $\operatorname{polylog}(n)$-bit encodings.
	In~the~\cc{QRAQM} model, a~dynamic set $D\subseteq\mathcal U$ of~at~most $M$ elements can be~stored in~$\widetilde O(M)$ space and constructed in~$\widetilde O(M)$ time.
	The~structure supports coherent successor lookup, insertion, and deletion in~$\widetilde O(1)$ time, up~to~inverse-polynomial error.
\end{theorem}

\paragraph{Quantum search.}
\begin{theorem}[\cite{Grover96,BHMT00,HMW03}]
\label{thm:quantum-search}
	Let $f:[N]\to\{0,1\}$ be~a~predicate with a~bounded-error coherent evaluation algorithm of~cost $T$.
	There is~a~bounded-error quantum algorithm that finds an~$x$ satisfying $f(x)=1$, or~reports that none exists, in~time $\widetilde O(\sqrt N\,T)$.
\end{theorem}

\paragraph{Quantum walks.}
Let $P$ be~a~reversible ergodic Markov chain on~a~finite state space, with stationary distribution $\pi$ and spectral gap $\delta$.
Fix a~marked set $\mathcal M$ and let $\mathcal D(x)$ denote the~data stored with a~state~$x$.
Let $T_{\operatorname{set}}$, $T_{\operatorname{upd}}$, and $T_{\operatorname{chk}}$ denote the~costs of~preparing the~stationary state and its data, coherently performing one transition and updating the~data, and checking the~vertex, respectively.
The~last operation maps, up~to~inverse-polynomial error,
\[\lvert x\rangle\lvert\mathcal D(x)\rangle\longmapsto(-1)^{[x\in\mathcal M]}\lvert x\rangle\lvert\mathcal D(x)\rangle,\]
and returns all ancillary registers to~$\lvert0\rangle$.

\begin{theorem}[\cite{Szegedy04,MNRS11}]
\label{thm:quantum-walk-search}
	If~$\pi(\mathcal M)\geq\mu$, then a~state in~$\mathcal M$ can be~found with bounded error in~time
	\[\widetilde O\left(T_{\operatorname{set}}+\frac1{\sqrt\mu}\left(\frac1{\sqrt\delta}T_{\operatorname{upd}}+T_{\operatorname{chk}}\right)\right).\]
\end{theorem}

\paragraph{Quantum walks on~Johnson graphs.}
For a~finite set $\mathcal X$ and $1\leq m\leq|\mathcal X|/2$, the~Johnson graph $J(\mathcal X,m)$ has the~$m$-subsets of~$\mathcal X$ as~its vertices, with two subsets adjacent when their intersection has size $m - 1$.
Our algorithm uses the~following lazy walk on~$J(\mathcal X,m)\times J(\mathcal X,m)$.
At~each step, the~walk successively samples one of~the~two component subsets, a~member of~that subset, and an~element of~$\mathcal X$ from their respective uniform distributions.
It~replaces the~sampled member by~the~sampled element unless the~latter already belongs to~the~subset, in~which case it~stays in~the same vertex.
This walk is~reversible and has the~uniform stationary distribution.
Kachigar and Tillich~\cite{KT17} show that its spectral gap is~$\delta=\Omega(1/m)$.
Consequently, if~at~least a~$\mu$ fraction of~its vertices are marked, \Cref{thm:quantum-walk-search} specializes to~the~running time
\[\widetilde O\left(T_{\operatorname{set}}+\frac1{\sqrt\mu}\left(\sqrt m\,T_{\operatorname{upd}}+T_{\operatorname{chk}}\right)\right).\]

\paragraph{Claw Finding.}
\begin{theorem}[\cite{Tani09}]
\label{thm:claw-finding}
	Let $\mathcal L$, $\mathcal R$, and $\mathcal Y$ be~finite sets with $|\mathcal Y|=\operatorname{poly}(|\mathcal{L}|, |\mathcal{R}|)$, and let $f:\mathcal L\to\mathcal Y$ and $g:\mathcal R\to\mathcal Y$.
	A~\emph{claw} is~a~pair $(x,y)\in\mathcal L\times\mathcal R$ satisfying $f(x)=g(y)$.
	Suppose that $|\mathcal L|\leq|\mathcal R|^2$, $|\mathcal R|\leq|\mathcal L|^2$, and each function can be~evaluated coherently in~$\log^{O(1)}(|\mathcal L|+|\mathcal R|)$ time.
	Then there is~a~bounded-error quantum algorithm that finds a~claw, or~reports that none exists, in~time
	\[\widetilde O\left(\left(|\mathcal L|\cdot|\mathcal R|\right)^{1/3}\right).\]
\end{theorem}

\section{\texorpdfstring{The $k$-SUM Algorithm}{The k-SUM Algorithm}}
\label{sec:k-sum}

We~are now ready to~prove~\Cref{thm:four-block}.

\thmmain*

\subsection{Partition}
Let $k_1,k_2,k_3,k_4$ be~positive integer parameters satisfying $k_1+k_2+k_3+k_4=k$; their values will be~chosen at~the end of~the~proof.
Group the~$k$ input lists into four consecutive blocks of~sizes $(k_1,k_2,k_3,k_4)$, numbered $1$, $2$, $3$, and $4$.
Thus, block~$1$ consists of~the~first $k_1$ input lists, and blocks $2$, $3$, and $4$ consist of~the~next $k_2$, $k_3$, and $k_4$ input lists, respectively.

The~index tuples of~the~four blocks range over the~product domains
\[\mathcal X_i=[n]^{k_i}\qquad\text{for every }i\in[4],\]
and we~write $\alpha_i\in\mathcal X_i$ for an~index tuple of~block $i$.
For each $i\in[4]$, let $\Sigma_i:\mathcal X_i\to\mathbb Z$ map a~tuple of~block $i$ to~its sum.
For example, if~$\alpha_1\in\mathcal X_1$, then $\Sigma_1(\alpha_1)=\sum_{j \in [k_1]} A_j[\alpha_1[j]]$.

\subsection{Construction}

\paragraph{Searching in~a~random subspace.}
The~straightforward approach would construct complete lists of~sums for the~blocks used in~the~search.
Instead, we~retain only a~small sample of~tuples from some of~the blocks and use a~quantum walk to~find the~appropriate subset of~tuples.

Let $r\in(0,\min\{k_2,k_4\}]$ be~a~parameter, whose value is~fixed at~the end of~the~proof, and set
\[m=\lfloor n^r/2\rfloor.\]
We~perform a~quantum walk on~the~graph
\[\mathcal G=J(\mathcal X_2,m)\times J(\mathcal X_4,m),\]
whose vertices are of~the~form
\[v=(S_2,S_4)\in\binom{\mathcal X_2}{m}\times\binom{\mathcal X_4}{m}.\]
We~will choose $k_2=k_4$, so~$|\mathcal X_2|=|\mathcal X_4|$ and $\mathcal G$ is~isomorphic to~$J(\mathcal X,m)\times J(\mathcal X,m)$ for a~set $\mathcal X$ of~that cardinality.
Each $\mathcal X_i$ is~the~set of~all index tuples for block~$i$; for $i\in\{2,4\}$, $S_i\subseteq\mathcal X_i$ is~the~$m$-element sample stored at~the~vertex, and $x\in S_i$ denotes a~single sampled tuple.
One step of~the walk chooses $i\in_R\{2,4\}$, $x\in_R S_i$, and $x'\in_R\mathcal X_i$.
If~$x'\notin S_i$, the~step replaces $x$ by~$x'$ in~that coordinate; otherwise it~remains at~$v$.

Fix a~particular $k$-$\problemabbr{SUM}$ solution and denote its four tuples of~selected indices, one for each block, by~$\alpha^*=(\alpha_1^*,\alpha_2^*,\alpha_3^*,\alpha_4^*)$.
For $i\in\{2,4\}$, a~uniformly random $m$-subset of~$\mathcal X_i$ contains a~fixed tuple with probability $m/n^{k_i}$.
Consequently, the~fraction of~vertices $v=(S_2,S_4)$ with $\alpha_2^*\in S_2$ and $\alpha_4^*\in S_4$ is
\[\Pr_{v = (S_2, S_4)\in_R V(\mathcal G)}\!\left[\alpha_2^*\in S_2\text{ and }\alpha_4^*\in S_4\right]=\frac{m^2}{n^{k_2+k_4}}.\]

\paragraph{Filtering by~a~random prime.}
Let $\mathcal P$ be~the~set of~primes in~$\{m\log^2 n, \dotsc, 2m\log^2 n\}$.
We~sample $p\in_R\mathcal P$ and keep it~fixed throughout the~walk.
For every fixed choice of~$p$, the~functions defined below determine a~fixed marked subset of~$V(\mathcal G)$.
When checking a~vertex, we~use quantum search over residues $q\in\mathbb Z_p$.
Once $q$ is~fixed, we~retain only the~left- and right-hand tuple combinations whose partial sums satisfy $\Sigma_1(\alpha_1)+\Sigma_2(\alpha_2)\equiv q\bmod p$ and $\Sigma_3(\alpha_3)+\Sigma_4(\alpha_4)\equiv-q\bmod p$, respectively.

\paragraph{Data stored at~a~vertex.}
At a~vertex $v = (S_2, S_4)$, the walk stores, for each $i \in \{ 2, 4 \}$, a~dictionary $D_{v, i}$ whose bucket indices are residues $s \in \mathbb{Z}_{p}$.
The bucket indexed by~$s$ is
\[D_{v,i}[s]=\{ (\Sigma_i(x), x)  \colon x \in S_i \text{ and } \Sigma_i(x)\equiv s\bmod p \}.\]
By~\Cref{thm:coherent-dictionary}, those dictionaries can be~built in~$\widetilde O(m)$ time, and each of~coherent lookup, insertion, and deletion costs $\widetilde O(1)$ time.
The~setup procedure prepares both dictionaries as~part of~the~walk state.
If~a~transition replaces $x$ by~$x'$ in~coordinate~$i$, its update coherently transforms the~vertex data from $(D_{v,2},D_{v,4})$ to~$(D_{v',2},D_{v',4})$ by~deleting $(\Sigma_i(x)\bmod p,(\Sigma_i(x),x))$ and inserting $(\Sigma_i(x')\bmod p,(\Sigma_i(x'),x'))$.

Define the~lookup function $R_{v,i}:\mathbb Z_p\to(\mathbb Z\times\mathcal X_i)\cup\{\bot\}$, where
\[R_{v,i}(s)=\begin{cases}\min D_{v,i}[s],&D_{v,i}[s]\neq\varnothing,\\ \bot,&D_{v,i}[s]=\varnothing.\end{cases}\]
Here, the~minimum is~taken lexicographically, first by~sum and then by~tuple.
If~$R_{v,i}(s)=(z,x)$, we~call $(z,x)$ the~\emph{canonical representative} of~the bucket indexed by~$s$.
Later, we~show that, with high probability, every entry in~the bucket containing a~solution has the same sum, so~the choice of~representative is~immaterial.

\paragraph{Checking a~vertex.}
Let $\alpha_1\in\mathcal X_1$ and $\alpha_3\in\mathcal X_3$.
Fix $q\in\mathbb Z_p$ and define the~two target bucket indices by
\[s_2^q(\alpha_1)=\bigl(q-\Sigma_1(\alpha_1)\bigr)\bmod p,\qquad s_4^q(\alpha_3)=\bigl(-q-\Sigma_3(\alpha_3)\bigr)\bmod p.\]
Thus, for $i\in\{2,4\}$, $s_i^q$ is~the~bucket index required for the~block-$i$ sum.
Define $F_{v,q}$ and $G_{v,q}$ by
\[F_{v,q}(\alpha_1)=\begin{cases}\Sigma_1(\alpha_1)+z,&R_{v,2}(s_2^q(\alpha_1))=(z,\alpha_2),\\ \bot,&R_{v,2}(s_2^q(\alpha_1))=\bot,\end{cases}\]
and
\[G_{v,q}(\alpha_3)=\begin{cases}-\Sigma_3(\alpha_3)-z,&R_{v,4}(s_4^q(\alpha_3))=(z,\alpha_4),\\ \top,&R_{v,4}(s_4^q(\alpha_3))=\bot.\end{cases}\]
Here, $\bot$ and $\top$ are two distinct symbols outside $\mathbb Z$.
If~$F_{v,q}(\alpha_1)=G_{v,q}(\alpha_3)$, then the~corresponding representatives satisfy $\Sigma_1(\alpha_1)+\Sigma_2(\alpha_2)=-\Sigma_3(\alpha_3)-\Sigma_4(\alpha_4)$, so~the~four tuples form a~solution.
Define the~marked set
\[\mathcal M_p=\{v\in V(\mathcal G):\exists q\in\mathbb Z_p,\ \alpha_1\in\mathcal X_1,\ \alpha_3\in\mathcal X_3\text{ such that }F_{v,q}(\alpha_1)=G_{v,q}(\alpha_3)\}.\]
The~checking procedure coherently evaluates the~predicate $v\in\mathcal M_p$ by~quantum search over~$q$ and \problemname{Claw Finding} on~$F_{v,q}$ and $G_{v,q}$, and then applies the~transformation
\[\lvert v\rangle\lvert D_{v,2},D_{v,4}\rangle\longmapsto(-1)^{[v\in\mathcal M_p]}\lvert v\rangle\lvert D_{v,2},D_{v,4}\rangle.\]
The~pseudocode is~presented in~\Cref{alg:four-block-overview}.

\algdef{SE}[WALK]{WalkBody}{EndWalkBody}{}{}
\algtext*{WalkBody}
\algtext*{EndWalkBody}

\begin{algorithm}[H]
\caption{Quantum $k$-\problemabbr{SUM} Algorithm}
\label{alg:four-block-overview}
\small
\begin{algorithmic}[1]
	\Statex \textbf{Input:} Integer lists $A_1,\ldots,A_k$, block sizes $k_1,k_2,k_3,k_4$, and a~parameter $r\in(0,\min\{k_2,k_4\}]$.
	\Statex \textbf{Output:} A~tuple of~indices $(i_1,\ldots,i_k)\in[n]^k$ satisfying $\sum_{j=1}^k A_j[i_j]=0$, if~one exists, and \textsc{no} otherwise.
	\State Partition the~lists into four blocks of~sizes $(k_1,k_2,k_3,k_4)$
	\State Set $m\gets\lfloor n^r/2\rfloor$ and choose a~prime $p\in_R\mathcal P$
	\State Run a~quantum walk over $\mathcal G$:
	\WalkBody
		\State Quantum search over $q \in \mathbb{Z}_{p}$
		\State Construct $F_{v, q}$ and $G_{v, q}$ and run \problemname{Claw Finding} on~them
		\If{a~claw is~found}
			\State \Return the~indices corresponding to~the~found claw
		\EndIf
	\EndWalkBody
	\State \Return \textsc{no}
\end{algorithmic}
\end{algorithm}

\subsection{Correctness}
The~algorithm cannot output an~incorrect solution; so~it~remains to~show that, whenever a~solution exists, the~algorithm finds one with constant probability.
Fix a~solution $\alpha^*=(\alpha_1^*,\alpha_2^*,\alpha_3^*,\alpha_4^*)$, and let $z_2^*=\Sigma_2(\alpha_2^*)$ and $z_4^*=\Sigma_4(\alpha_4^*)$.
We~say that $z_2^*$ is~\emph{exposed} at~a~vertex if~the~canonical representative of~its residue has exact sum $z_2^*$, and define exposure of~$z_4^*$ analogously.

Let $\mathcal V^*$ be~the~set of~vertices whose two sampled subsets contain $\alpha_2^*$ and~$\alpha_4^*$.
For a~fixed prime $p$, call a~vertex $v\in\mathcal V^*$ \emph{successful} if~both target sums $z_2^*$ and~$z_4^*$ are exposed at~$v$.
Call $p$ \emph{good} if~at~least half of~the~vertices in~$\mathcal V^*$ are successful.

\begin{lemma}
\label{lem:good-prime}
	A~prime $p\in_R\mathcal P$ is~good with constant probability.
\end{lemma}
\begin{proof}
	Fix a~vertex $v\in\mathcal V^*$.
	For $i\in\{2,4\}$, an~entry $x\in S_i$ with $\Sigma_i(x)\neq z_i^*$ can lie in~the~same residue bucket as~$z_i^*$ only if~$p$ divides $\Sigma_i(x)-z_i^*$.
	The~difference $\Sigma_i(x)-z_i^*$ is~nonzero and has absolute value $\widetilde O(n^k)$, therefore applying \Cref{prop:prime-divisor} bounds the~collision probability by
	\[
	O\left(\frac{\log \widetilde{O}(n^{k}) }{m \log^2 n}\right) = O\left(\frac{k}{m \log n}\right) = O\left(\frac{1}{m \log n}\right),
	\] 
	where the~last equality uses that $k$ is~fixed.
	A~union bound over at~most $2m$ tuples in~the~two sampled subsets shows that, except with probability $O(1/\log n)$, both relevant buckets contain only entries of~exact sums $z_2^*$ and $z_4^*$, respectively; in~particular, $v$ is~successful.
	So, the~expected fraction of~unsuccessful vertices of~$\mathcal V^*$ is~$O(1/\log n)$.
	Markov's inequality now implies that more than half of~the~vertices of~$\mathcal V^*$ are unsuccessful with probability $O(1/\log n)$.
\end{proof}

Every successful vertex belongs to~$\mathcal M_p$, hence if~$p$ is~good, at~least half of~$\mathcal V^*$ is~contained in~$\mathcal M_p$, and then
\begin{equation}\label{eq:marked-fraction}\frac{|\mathcal M_p|}{|V(\mathcal G)|}\geq\frac12\frac{|\mathcal V^*|}{|V(\mathcal G)|}=\frac12\frac{m^2}{n^{k_2+k_4}}=\Omega(n^{2r-k_2-k_4}).\end{equation}

By~\Cref{lem:good-prime}, the~bound in~\Cref{eq:marked-fraction} holds with constant probability over the~choice of~$p$.
The~residue search and the claw-finding procedure are amplified to~error $n^{-\omega(1)}$, as~required by~\Cref{thm:quantum-walk-search}.
Thus the~algorithm finds a~solution with constant probability whenever one exists.

\subsection{Running Time}
By~\Cref{lem:good-prime,eq:marked-fraction}, with constant probability over $p\in_R\mathcal P$, the~quantum walk has a~marked fraction $\mu=\Omega(n^{2r-k_2-k_4})$.
The~quantum search over $q\in\mathbb Z_p$ contributes a~factor $\widetilde O(\sqrt p)$ to~the~checking time by~\Cref{thm:quantum-search}.
For each fixed residue, \problemname{Claw Finding} searches the~domains $\mathcal X_1$ and $\mathcal X_3$.
For the~parameter choices made below, we~force $k_1/2\le k_3\le 2k_1$, so~\Cref{thm:claw-finding} gives a~claw finding of~a~fixed residue in~time $\widetilde O(n^{(k_1+k_3)/3})$.
By~\Cref{thm:coherent-dictionary}, the~walk setup builds the~two dictionaries in~$\widetilde O(m)=\widetilde O(n^r)$ time.
Since adjacent vertices differ by~one tuple, a~walk update maintains them using one deletion and one insertion in~$\widetilde O(1)$ time.
Using~\Cref{thm:quantum-walk-search} on~$\mathcal{G}$ with parameters
\[\begin{aligned}T_{\operatorname{set}}&=\widetilde O(n^r),\qquad&T_{\operatorname{upd}}&=\widetilde O(1),\\T_{\operatorname{chk}}&=\widetilde O(\sqrt p\,n^{(k_1+k_3)/3})=\widetilde O(n^{(k_1+k_3)/3+r/2}),\qquad&\mu&=\Omega(n^{2r-k_2-k_4})\end{aligned}\]
gives a~total running time of~the algorithm
\[\widetilde O\left(n^r+n^{(k_2+k_4)/2-r}\left(n^{r/2}+n^{(k_1+k_3)/3+r/2}\right)\right)=\widetilde O\left(n^{\max\left\{r,\ \frac{k_1+k_3}{3}+\frac{k_2+k_4-r}{2}\right\}}\right).\]

It~remains to~choose the~block sizes and~$r$.
The~bound depends on~$k_2$ and $k_4$ only through their sum, except for the~constraint $r\le\min\{k_2,k_4\}$; hence we~choose $k_2=k_4$.
Similarly, it~depends on~$k_1$ and $k_3$ only through their sum, and choosing them as~evenly as~possible ensures the~conditions of~\Cref{thm:claw-finding}.
For convenience in~\Cref{tab:params}, put $\ell=k_1+k_3$; then $k_1=\lfloor\ell/2\rfloor$, $k_3=\lceil\ell/2\rceil$, and $k=\ell+2k_2$.
With these choices, put
\[r=\min\left\{k_2,\frac{2}{3}\left(k_2+\frac{k_1+k_3}{3}\right)\right\}.\]

Let $k=7d+j$, where $j\in\{0,1,\ldots,6\}$.

\begin{table}[H]
\centering
\begin{tabular}{ccccc}
\toprule
$j = k \bmod 7$ & $\ell = k_1 + k_3$ & $k_2=k_4$ & $r$ & $\Psi_k = \max \{ r, \ell / 3 + k_2 - r / 2 \}$\\
\midrule
$0$ & $3d$ & $2d$ & $2d$ & $2d$\\
$1$ & $3d+1$ & $2d$ & $2d$ & $2d+1/3$\\
$2$ & $3d$ & $2d+1$ & $2d+2/3$ & $2d+2/3$\\
$3$ & $3d+1$ & $2d+1$ & $2d+8/9$ & $2d+8/9$\\
$4$ & $3d+2$ & $2d+1$ & $2d+1$ & $2d+7/6$\\
$5$ & $3d+3$ & $2d+1$ & $2d+1$ & $2d+3/2$\\
$6$ & $3d+2$ & $2d+2$ & $2d+16/9$ & $2d+16/9$\\
\bottomrule
\end{tabular}
\caption{Optimal parameters for~\Cref{alg:four-block-overview}, where $d=\lfloor k/7\rfloor$ and $j=k\bmod 7$.}
\label{tab:params}
\end{table}

This completes the~proof of~\Cref{thm:four-block}.

\section{Application to~Pigeonhole Modular Equal Subset Sum}
\label{sec:pigeonhole-modular-equal-subset-sum}

\begin{definition}
	An~instance of~\problemname{Pigeonhole Modular Equal Subset Sum} consists of~a~list $A$ of~$n$ positive integers and a~modulus $q\leq2^n-1$.
	The~task is~to~find two distinct subsets $S_1,S_2\subseteq[n]$ such that
		\[\sum_{i\in S_1}A[i]\equiv\sum_{i\in S_2}A[i]\bmod q.\]
\end{definition}

\begin{theorem}
\label{thm:pigeonhole-modular-equal-sums}
	The~\problemname{Pigeonhole Modular Equal Subset Sum} problem can be~solved in~quantum time~$O^*(3^{2n/7})$.
\end{theorem}
\begin{proof}
	Replacing every $A[i]$ by~$A[i]\bmod q$ does not change the set of~solutions, so~we~assume that $0\leq A[i]<q$ for all~$i\in[n]$.

	A~solution is~equivalent to~a~nonzero vector $\varepsilon\in\{-1,0,1\}^n$ satisfying
	\[\sum_{i \in  [n]}\varepsilon[i]\,A[i]\equiv0\bmod q.\]
	As~$0\leq A[i]<q$ for all~$i$, the~left-hand side lies in~the~interval $(-nq,nq)$, so~the~congruence holds if~and only~if
	\[\sum_{i \in [n]}\varepsilon[i]\,A[i]=cq\quad\text{for some integer }c\text{ with }|c|<n.\]

	Partition $[n]$ into seven blocks $B_1,\ldots,B_7$, each of~size at~most $\lceil n/7\rceil$.
	For each $j\in[7]$, construct a~list $L_j$ indexed by~vectors $\varepsilon\in\{-1,0,1\}^{B_j}$, with $L_j[\varepsilon]=\sum_{i\in B_j}\varepsilon[i]\,A[i]$.
	Each list has size at~most $M=3^{\lceil n/7\rceil}$.
	For a~fixed~$c$, finding a~vector $\varepsilon\in\{-1,0,1\}^n$ with $\sum_{i \in [n]}\varepsilon[i]\,A[i]=cq$ is~exactly a~$7$-\problemabbr{SUM} instance on~the~lists $L_1,\ldots,L_7$ with target~$cq$.

	For $c = 0$, one may get the zero vector as~an~answer, so~we~need to~force at~least one coordinate of~$\varepsilon$ to~be~nonzero.
	For $i\in[n]$, $\sigma\in\{-1,1\}$, and an~integer~$c$ with $|c|<n$, let $j(i)$ denote the index of~the~block containing~$i$, and let $I_{i,\sigma,c}$ be~the~$7$-\problemabbr{SUM} instance with target~$cq$ obtained from $L_1,\ldots,L_7$ by~restricting $L_{j(i)}$ to~the~indices $\varepsilon\in\{-1,0,1\}^{B_{j(i)}}$ with $\varepsilon[i]=\sigma$.
	Every solution of~$I_{i,\sigma,c}$ satisfies $\varepsilon[i]=\sigma$ and is~therefore a~nonzero solution of~the~displayed congruence.
	By~\Cref{thm:four-block} with $k=7$, each of~them is~solved in~time $\widetilde O(M^2)$, so~the~total running time is
	\[O(n^2)\cdot\widetilde O(M^2)=O^*(3^{2n/7}).\qedhere\]
\end{proof}

\section*{Acknowledgments}
The~authors used ChatGPT to~assist with verifying all proofs, identifying typos, and conducting the literature review.
The~authors assume responsibility for all content.

\bibliographystyle{alpha}
\bibliography{refs}

\end{document}